\documentclass[a4paper,final]{llncs}
\pdfoutput=1

\usepackage[utf8]{inputenc}
\usepackage[T1]{fontenc}
\usepackage{amsmath}
\usepackage{amssymb}
\usepackage{xspace}
\usepackage{ifdraft}%
\usepackage{proofs}
\usepackage{proof}
\usepackage{color}
\usepackage{hyperref}
\usepackage{tikz}
\usepackage{wrapfig}
\hypersetup{
  pdfauthor = {Maciej Zielenkiewicz and Aleksy Schubert},
  pdftitle = {Automata Theory Approach to Predicate Intuitionistic~Logic},
  pdfsubject = {proof synthesis, program synthesis},
  pdfkeywords = {},
  pdfcreator = {LaTeX with hyperref package},
  pdfproducer = {pdflatex},
  bookmarks=true,
  final
}

\def\letBeIn#1#2#3{\textrm{\underline{let} }#1\textrm{ \underline{be} }#2\textrm{ \underline{in} }#3\xspace}
\def\packTo#1#2#3#4{\textrm{\underline{pack} }#1,\,#2\textrm{ \underline{to} }\exists#3.\,#4}

\def\stateno#1{\mathfrak #1}
\def\stateid#1#2{q^{#1}_{\stateno{#2}}}

\DeclareMathOperator{\inlOp}{\underline{in}_1}
\DeclareMathOperator{\inrOp}{\underline{in}_2}
\def\inl#1#2{\inlOp_{#1}#2}
\def\inr#1#2{\inrOp_{#1}#2}

\newcommand{\instruct}[2]{\mathsf{#1\ }#2}
\newcommand{\new}[1]{\instruct{new}{#1}}
\newcommand{\jmp}[1]{\instruct{jmp}{#1}}
\def\store#1{\instruct{store}{#1}}
\def\check#1{\instruct{check}{#1}}
\def\instr#1{\instruct{instR}{#1}}
\def\instl#1{\instruct{instL}{#1}}
\def\load#1{\instruct{load}{#1}}
\def\caseOf#1#2#3#4#5{\textrm{\underline{case} }#1 \textrm{
    \underline{of} }\left[#2\right]#3,\,\left[#4\right]#5}

\def\iline#1#2#3{$\stateid{#1}{#2}\colon$& $#3$ &}
\def\ixline#1#2#3#4{$q^{#1}_{\stateno #2,#3}\colon$& $#4$ &}
\def\ailine#1#2{$q^{#1}_a\colon$& $#2$& }
\def\aixline#1#2#3{$q^{#1}_{a,#2}\colon$& $#3$&}
\newcommand{\arity}[1]{\mathrm{arity}(#1)\xspace}

\newcommand{\FVf}[1]{{\mathrm{FV}}_1(#1) }
\newcommand{\fvar}{{\cal X}_\mathrm{1}\xspace}
\newcommand{\pvar}{{\cal X}_\mathrm{p}\xspace}
\newcommand{\pr}[1]{\textsc{#1}\xspace}
\newcommand{\<}{\langle\xspace}
\renewcommand{\>}{\rangle\xspace}
\newcommand{\bbot}{\mbox{{$\bot\hspace{-1.5ex}\bot$}}\xspace}
\newcommand{\aut}[1]{\mathbb{#1}\xspace}
\newcommand{\parfunc}{\rightharpoonup\xspace}
\newcommand{\fv}{\mathrm{fv}\xspace}
\newcommand{\succc}{\;\mathrm{succ}\;\xspace}
\newcommand{\axiom}{\mathrm{axiom}\xspace}
\newcommand{\bind}{\mathsf{bind}\xspace}

\title{Automata Theory Approach to Predicate Intuitionistic~Logic}
\author{Maciej Zielenkiewicz and Aleksy Schubert}
\institute{Institute of Informatics, University of Warsaw\\
   ul. S. Banacha 2, 02--097 Warsaw, Poland\\
\email{\tt [maciekz,alx]@mimuw.edu.pl}}

\begin{document}

\maketitle

\begin{abstract} 
  Predicate intuitionistic logic is a well established fragment of
  dependent types. According to the Curry-Howard isomorphism proof
  construction in the logic corresponds well to synthesis of a program
  the type of which is a given formula.  We present a model of
  automata that can handle proof construction in full intuitionistic
  first-order logic. The automata are constructed in such a way that
  any successful run corresponds directly to a normal proof in the
  logic. This makes it possible to discuss formal languages of proofs
  or programs, the closure properties of the automata and their
  connections with the traditional logical connectives.
\end{abstract}

\section{Introduction}

Investigations in automata theory lead to abstraction of algorithmic
processes of various kinds. This enables analysis of their strength
both in terms of their expressibility (i.e.\ answer questions on which
problems can be solved with their help) and in terms of resources they
consume (i.e.\ time or space). They also make it possible to shed a
different light on the original problem (e.g.\ the linguistic problem
of languages generated by grammars can be reduced to the analysis of
pushdown automata) which makes it possible to conduct analysis that
was not possible before. In addition, the automata become a particular
compact data structure that can in itself, when defined formally, be
subject to further computations, as finite or pushdown automata are in
automata theory.

Typically, design of automata requires extraction of finite control
over the process of interest. This is not always immediate in
$\lambda$-calculi as $\lambda$-terms can contain bound variables from
an infinite set. One possibility here consist in
restricting the programming language so that there is no need to
introduce binders. This method was used in the work of Düdder et al
\cite{DudderMJ14}, which was enough to synthesise $\lambda$-terms that
were programs in a simple but expressive functional language.

Another approach would be to restrict the program search to programs
in \emph{total discharge form}. In programs of this form, one needs to
keep track of types of available library calls, but not of the call names
themselves.  This idea was explored by Takahashi et al
\cite{TakahashiAH96} who defined context-free grammars that can be
used for proof search in propositional intuitionistic logic, which is,
by Curry-Howard isomorphism, equivalent to program search in the simply
typed $\lambda$-calculus. Actually, the grammars can be viewed as
performing program search by means of tree automata due to known
correspondence between grammars and tree automata. However, the
limitation to total discharge form can be avoided with help of the
technique developed by Schubert, Dekkers and Barendregt
\cite{SchubertDB15}.

A different approach to abstract machinery behind program search
process was proposed by Broda and Damas \cite{BrodaD05} who developed
a formula-tree proof method. This technique provides a realisation of
the proof search procedure for a particular propositional formula as a
data structure, which can be further subject to algorithmic
manipulation.

In addition to these investigations for intuitionistic propositional
logic there was a proposal to apply automata theoretic notions to
proof search in first-order logic \cite{Hetzl2012}. The paper
characterises a class of proofs in intuitionistic first-order logic
with so called \emph{tree automata with global equalities and
  disequalities} (TAGED) \cite{FiliotTT10}. The characterisation makes
it possible to recognise proofs that are not necessarily in normal
form, but is also limited to certain class of tautologies (as the
emptiness problem for the automata is decidable).

In this paper we propose an automata theoretical abstraction of the
proving process in full intuitionistic first-order logic. Its
advantages can be best expressed in terms in which implicit, but
crucial, features of proof search become explicit. In our automata the
following elements of the proving process are exposed.
\begin{itemize}
\item The finite control of the proving process is made explicit.
\item A binary internal structure of the control is explicated where
  one component corresponds to a subformula of the original formula
  and one to the internal operations that should be done to handle the
  proof part relevant for the subformula. As a by-product of this
  formulation it becomes aparent how crucial role the subformula
  property plays in the proving process.
\item The resource that serves to represent eigenvariables that occur
  in the process is distinguished. This abstraction is important as
  the variables play crucial role in complexity results concerning the
  logic \cite{SchubertUZ2015,RPQNisNE}.
\item The automata enable the possibility of getting rid of the particular
  syntactical form of formulas and instead work on more abstract
  structures.
\item The definition of automaton distils the basic instructions
  necessary to conduct the proof process, which brings into the view
  more elementary operations the proving process depends on. 
\end{itemize}
Although the work is formulated in terms of 
logic, it can be viewed as synthesis of programs in a
restricted class of dependently typed functional programs.

\paragraph{Organisation of the paper} We fix the notation and present
intuitionistic first-order logic in
Section~\ref{sec:preliminaries}. Next, we define our automata in
Section~\ref{sec:automata}.  We summarise the account in
Section~\ref{sec:conclusions}.

\section{Preliminaries}
\label{sec:preliminaries}

We need to fix the notation and present the basic facts about
intuitionistic first-order logic. The notation $A\parfunc B$ is used
to denote the type of partial functions from $A$ to $B$. We write
$\dom{w}$ for the domain of the function $w:A\parfunc B$.  For two
partial functions $w, w'$ we define
$w\oplus w' = w \cup \{\<x,y\>\in w' \mid x\not\in\dom{w}\}$.  The set
of all subsets of a set $A$ is $P(A)$.

A prefix closed set of strings $\Nat^*$ over $\Nat$ is called a
\emph{carrier of a tree}. A tree is a tuple $\<A,\leq,L,l\>$ where
$A$ is a carrier of the tree, $\leq$ is the prefix order on $\Nat^*$,
the set $L$ is the set of labels and $l:A\to L$ is the \emph{labelling
  function}. Whenever the set of labels and the labelling function
are clear from the context, we abbreviate the quadruple to the tuple
$\<A,\leq\>$. Since the formula notation makes it easy, we sometimes
use a subtree $\varphi$ of $A$ to actually denote a node in $A$ at
which $\varphi$ starts.

\subsection{Intuitionistic First-Order Logic}
\label{sec:first-order-logic}

The basis for our study is the first-order intuitionistic logic (for
more details see e.g.\ the work of Urzyczyn, \cite{Urzyczyn2016}). We
assume that we have a set of predicates ${\cal P}$ that can be used to
form atomic formulae and an infinite set $\fvar$ of first-order
variables, usually noted as $X, Y, Z$ etc.\ with possible
annotations. Each element $\pr{P}$ of ${\cal P}$ has an arity, denoted
$\arity{\pr{P}}$.  The formulae of the system are:
$$
\varphi,\psi::= \pr{P}(X_1,\ldots,X_n) \mid
                \varphi_1\land\varphi_2 \mid
                \varphi_1\lor\varphi_2 \mid
                \varphi_1\to \varphi_2 \mid 
                \forall X.\varphi \mid
                \exists X.\varphi \mid 
                \bot
$$
where $\pr{P}$ is an $n$-ary predicate and $X,X_1,\ldots,X_n\in\fvar$.
We follow Prawitz and introduce negation as a notation defined
$\lnot\varphi = \varphi\to\bot$.  A formula of the form
$\pr{P}(X_1,\ldots,X_n)$ is called an \emph{atom}.  A \emph{pseudo-atom
  formula} is a formula of one of the three forms: atom formula, a formula
of the form $\exists X.\varphi$, or a formula of the form
$\varphi_1\lor\varphi_2$.
We do not include parentheses in the grammar since we actually
understand the formulas as abstract syntax trees instead of strings.
The tree is traditionally labelled with the cases of the above
mentioned grammar. We assume that for a given case in the grammar the
corresponding node of the tree has as many sons as there are
non-terminal symbols in the case.  In addition, we use in writing
traditional disambiguation conventions for $\land, \lor$ and insert
parentheses to further disambiguate whenever this is
necessary. The connective $\to$ is understood as
right-associative so that $\varphi_1\to \varphi_2\to \varphi_3$ is
equivalent to $\varphi_1\to (\varphi_2\to \varphi_3)$. In a formula
$\varphi = \varphi_1\to\cdots\to \varphi_n\to \varphi'$, where
$\varphi'$ is a pseudo-atom, the formula $\varphi'$ is called
\emph{target of $\varphi$}. In case $\varphi' = \exists X.\varphi''$,
we call it \emph{existential target} of $\varphi$. 

The set of {\em free first-order
  variables} in a formula $\varphi$, written $\FVf{\varphi}$, is
\begin{itemize}
\item $\FVf{\pr{P}(X_1,\ldots, X_n)} = \{ X_1,\ldots, X_n\}$,
\item $\FVf{\varphi_1\ast \varphi_2} =
  \FVf{\varphi_1}\cup\FVf{\varphi_2}$ where $\ast\in\{\land,\lor,\to\},$
\item $\FVf{Q X.\varphi} = \FVf{\varphi}\backslash \{X\}$ where
  $Q\in\{\exists,\forall\}$,
\item $\FVf{\bot} = \emptyset$.
\end{itemize}
Other variables that occur in a formula are bound. Terms that differ
only in renaming of bound variables are $\alpha$-equivalent
and we do not distinguish between them.
To describe the binding structure of a formula we use a~special
$\bind$ operation.  Let us assume that a formula $\varphi$ has no free
variables (i.e.\ $\FVf{\varphi}=\emptyset$) and let $\psi$ be its
subformula together with a variable $X$ free in $\psi$. We define
$\bind_\varphi(\psi,X)$ as the subformula of $\varphi$ that binds the
free occurrences of $X$ in $\psi$, i.e.\ the subformula $\varphi'$ of
$\varphi$ such that each its proper subformula $\psi''$ that contains
$\psi$ as a subformula has $X\in\FV{\psi''}$. For instance
$\bind_{\bot\to\exists X.\bot\to P(X)}(P(X),X) = \exists X.\bot\to
P(X)$.

\begin{figure}[htb]
\begin{center}\small
$$
\infer[(var)]{\Gamma, x\!:\!\varphi\vdash x:\varphi}{}
$$
\vspace{-1ex}
$$
\infer[(\land I)]{\Gamma\vdash \<M_1, M_2\> :\varphi_1\land \varphi_2}
{\Gamma\vdash M_1:\varphi_1  &
\Gamma\vdash M_2:\varphi_2}
$$
$$
\infer[(\land E1)]{\Gamma\vdash \pi_1 M:\varphi_1}
{\Gamma\vdash M: \varphi_1\land \varphi_2}
\quad
\infer[(\land E2)]{\Gamma\vdash \pi_2M:\varphi_2}
{\Gamma\vdash M: \varphi_1\land \varphi_2}
$$
\vspace{-1ex}
$$
\infer[(\lor I1)]{\Gamma\vdash \inl{\varphi_1\lor\varphi_2}M :\varphi_1\lor \varphi_2}
{\Gamma\vdash M:\varphi_1}
\qquad
\infer[(\lor I1)]{\Gamma\vdash \inr{\varphi_1\lor\varphi_2}M :\varphi_1\lor \varphi_2}
{\Gamma\vdash M:\varphi_2}
$$
$$
\infer[(\lor E)]{\Gamma\vdash
  \caseOf{M}{x:\varphi_1}{N_1}{y:\varphi_2}{N_2}:\varphi}
{\Gamma\vdash M: \varphi_1\lor \varphi_2 & 
  \Gamma,x:\varphi_1\vdash N_1: \varphi &
  \Gamma,y:\varphi_2\vdash N_2: \varphi}
$$
\vspace{-1ex}
$$
\infer[(\to I)]{\Gamma\vdash\lambda x:\varphi_1.M:\varphi_1\to \varphi_2}
{\Gamma, x\!:\! \varphi_1\vdash M:\varphi_2}
\qquad
\infer[(\to E)]{\Gamma\vdash M_1M_2:\varphi_2}
{\Gamma\vdash M_1: \varphi_1\to \varphi_2 & \Gamma\vdash M_2: \varphi_1}
$$
\vspace{-1ex}
$$
\infer[(\forall I)^*]{\Gamma\vdash \lambda X M : \forall X.\varphi}
{\Gamma\vdash M:\varphi}
\qquad
\infer[(\forall E)^*]
{\Gamma\vdash MY:\varphi[X:=Y]}
{\Gamma\vdash M:\forall X.\varphi}
$$
\vspace{-1ex}
$$
\infer[(\exists I)]{\Gamma\!\vdash\! \packTo{\!M\!}{\!Y\!}{X}{\varphi} : \exists
  X.\varphi}%
{\Gamma\!\vdash M:\varphi[X:=Y]}
\quad
\infer[(\exists E)^*]
{\Gamma\!\vdash\! \letBeIn{x\!:\!\varphi}{M_1\!:\!\exists X.\varphi}{M_2}:\psi}
{\Gamma\!\vdash M_1:\exists X.\varphi & \Gamma, x\!:\!\varphi\vdash M_2:\psi}
$$
\infer[(\bot E)]
{\Gamma\vdash \bbot_\varphi M:\varphi}
{\Gamma\vdash M:\bot}
\end{center}

\rule{5em}{0.5pt}\\
${}^*$ Under the eigenvariable condition
$X\not\in FV(\Gamma, \psi)$.

  \caption{The rules of the intuitionistic first-order logic}
  \label{fig:rules}
\end{figure}

For the definition of \emph{proof terms} we assume that there is an
infinite set of \emph{proof term variables} $\pvar$, usually noted as
$x, y, z$ etc.\ with possible annotations. These can be used to form the
following terms.
$$
\begin{array}{l@{\;}l}
M,N ::= & x\mid 
          \<M_1, M_2\> \mid \pi_1 M\mid \pi_2 M \mid \\&
          \inl{\varphi_1\lor\varphi_2}M \mid
          \inr{\varphi_1\lor\varphi_2}M \mid 
          \caseOf{M}{x:\varphi_1}{N_1}{y:\varphi_2}{N_2} \mid \\ &
          \lambda x:\varphi.M \mid M_1M_2 \mid 
          \lambda X M \mid MX \mid \\ &
          \packTo{M}{Y}{X}{\varphi} \mid
          \letBeIn{x:\varphi}{M_1:\exists X.\varphi}{M_2} \mid \bbot_\varphi M
\end{array}
$$
where $x$ is a proof term variable, $\varphi, \varphi_1, \varphi_2$
are first-order formulas and $X, Y$ are first-order variables. Due to
Curry-Howard isomorphism the proof terms can serve as programs in a
functional programming language. Their operational semantics is given
in terms of reductions. Their full exposition can be found in the work
of de Groote \cite{deGroote02}. We omit it here, but give an intuitive
account of the meaning of the terms. In particular,
$\langle M_1, M_2\rangle$ represents the product aggregation construct
and $\pi_iM$ for $i=1,2$ decomposition of the aggregation by means of
projections. The terms $\inl{\varphi_1\lor\varphi_2}M$,
$\inr{\varphi_1\lor\varphi_2}M$ reinterpret the value of $M$ as one in
type $\varphi_1\lor\varphi_2$. At the same time
$\caseOf{M}{x:\varphi_1}{N_1}{y:\varphi_2}{N_2}$ construct offers the
possibility to make case analysis of a value in an $\lor$-type. This
construct is available in functional programming languages in a more
general form of algebraic types. The terms $\lambda x:\varphi.M$,
$M_1M_2$ represent traditional function abstraction and
application. The proof terms that represent universal quantifier
manipulation make it possible to parametrise type with a particular
value $\lambda X M$ and use the parametrised term for a particular
case $MX$. At last $\packTo{M}{Y}{X}{\varphi}$ makes it possible to
hide behind a variable $X$ an actual realisation of a construction
that uses another individual variable $Y$. The abstraction obtained in
this way can be used using
$\letBeIn{x:\varphi}{M_1:\exists X.\varphi}{M_2}$. At last the term
$\bbot_\varphi M$ corresponds to the break instruction.

The environments ($\Gamma, \Delta$
etc.\ with possible annotations) in the proving system are finite sets of pairs $x:\psi$
that assign formulas to proof variables. We write $\Gamma\vdash M:A$
to express that the judgement is indeed derivable.  The inference
rules of the logic are presented in~Fig.~\ref{fig:rules}. We have now
two kinds of free variables, namely free proof term variables and free
first-order variables. The set of free term variables is defined
inductively as follows
\begin{itemize}
\item $\FV{x} = \{x\}$,
\item $\FV{\<M_1, M_2\>} = \FV{M_1M_2} = \FV{M_1}\cup\FV{M_2}$,
\item $\FV{\pi_1 M} = \FV{\pi_2 M} = 
  \FV{\inl{\varphi_1\lor\varphi_2}M} =
  \FV{\inr{\varphi_1\lor\varphi_2}M} =$\\ 
  $\FV{\lambda X M} =  \FV{MX} = \FV{\packTo{M}{Y}{X}{\varphi}} = 
  \FV{\bbot_\varphi M} =$\\
  $\FV{M}$,
\item $\FV{\caseOf{M}{x:\varphi_1}{N_1}{y:\varphi_2}{N_2}} =$\\
 $ \FV{M}\cup(\FV{N_1}\backslash\{x\})\cup(\FV{N_2}\backslash y),$
\item $\FV{\lambda x:\varphi.M} = \FV{X}\backslash\{x\},$
\item
  $\FV{\letBeIn{x:\varphi}{M_1:\exists X.\varphi}{M_2}} =
  \FV{M_1}\cup(\FV{M_2}\backslash\{x\}).$
\end{itemize}
Again, the terms that differ only in names of bound term variables are
considered $\alpha$-equivalent and are not distinguished by us. Note
that we can use the notation $\FVf{M}$ to refer to all free type
variables that occur in $M$. This set is defined by recursion over the
terms and taking all the free first-order variables that occur in
formulas that are part of the terms so that for instance
$\FVf{\inl{\varphi_1\lor\varphi_2}M} =
\FVf{\varphi_1}\cup\FVf{\varphi_2}\cup\FVf{M}$. At the same time there
are naturally terms that bind first-order variables,
$\FVf{\lambda XM} = \FVf{M}\backslash\{X\}$ and bring new free
first-order ones, e.g.\ $\FVf{MX} = \FV{M}\cup\{X\}$.

Traditionally, the $(cut)$ rule is not mentioned among standard rules
in Fig.~\ref{fig:rules}, but as it is common in $\lambda$-calculi, it
is included it in the system in the form of a $\beta$-reduction rule.
This rule forms the basic computation mechanism in the stystem
understood as a programming language. We omit the rules due to the
lack of space, but an interested reader can find them in the work of
de~Groote \cite{deGroote02}. Still, we want to focus our attention to
terms in normal form (i.e.\ terms that cannot be further
reduced). Partly because the search for terms in such form is easier
and partly because source code of programs contains virtually
exclusively terms in normal form. The following theorem states that
this simplification does not make us lose any possible programs in our
program synthesis approach.
\begin{theorem}[Normalisation]
  \label{theorem:normalisation-first-order}
  First-order intuitionistic logic is strongly normalisable i.e.\ each
  reduction has a finite number of steps.
\end{theorem}
The paper by de~Groote contains also (implicitly) the following
result.
\begin{theorem}[Subject reduction]
  \label{theorem:subject-reduction-first-order}
  First-order intuitionistic logic has the subject reduction property,
  i.e.\ if $\Gamma\vdash M:\phi$ and $M\to_{\beta\cup p} N$ then
  $\Gamma\vdash N:\phi$.
\end{theorem}
As a consequence we obtain that each provable formula has a proof in
normal form. However, we need in our proofs a stricter notion of
\emph{long normal form}.

\subsection{Long normal forms}
We restrict our attention to terms which are in \emph{long normal form}.
The idea of long normal form for our logic is best explained by the following
example (\cite{Urzyczyn2016}, section 5): suppose $X\colon r$ and $Y\colon
r\to p\lor q$. The long normal form of $YX$ is 
$\caseOf{YX}{a:p}{\lambda u.\inl u}{b:q}{\lambda v.\inr v}$.

Our definitions follow those of Urzyczyn, \cite{Urzyczyn2016}. We classify
normal forms into:
\begin{itemize}
\item introductions $\lambda X.N$, $\lambda x.N$, $\langle N1, N2\rangle$, $\inl
N$, $\inr N$, $\packTo{N}{y}{X}{\varphi}$,
\item proper eliminators $X$, $P N$, $\pi_i P$, $P(x)$,
\item improper eliminators $\bbot_\varphi(P)$,
$\caseOf{P}{x:\varphi_1}{N_1}{y:\varphi_2}{N_2}$,\\ $\letBeIn{x:\varphi}{N:\exists
X.\varphi}{P}$
\end{itemize}
where $P$ is a proper eliminator and $N$ is a normal form. The long normal forms
(lnfs) are defined recursively with \emph{quasi-long proper eliminators}:
\begin{itemize}
\item A quasi-long proper eliminator is a~proper eliminator where all arguments are of 
pseudo-atom type.\footnote{Note that a variable is a quasi-long proper
eliminator because all arguments is an empty set in this case.}
\item A constructor $\lambda X.N$, $\langle N_1,N_2\rangle$, $\in_i N$,
$\underline{pack}\ldots$, $\underline{let}\ldots$ is a lnf when its
arguments are lnfs.
\item A case-eliminator $\caseOf{P}{x:\varphi_1}{N_1}{y:\varphi_2}{N_2}$ is a lnf
when $N_1$ and $N_2$ are lnfs and $P$ is a quasi-long
proper eliminator.  A miracle (\emph{ex falso quodlibet}) $\bbot_\varphi(P)$ of a target type $\tau$ is a
long normal form when $P$ is a quasi-long proper eliminator of type $\varphi$.
\item A eliminator $\letBeIn{x:\varphi}{N:\exists X.\varphi}{P}$ is a lnf when
$N$ is a lnf and $P$ is a~quasi-long proper eliminator.
\end{itemize}
The usefulnes of these forms results from the following proposition,
\cite{Urzyczyn2016}.

\begin{proposition}[Long normal forms]
  If $\Gamma\vdash M:\phi$ then there is a long normal form $N$ such that
  $\Gamma\vdash N:\phi$.
\end{proposition}

The design of automata that handle proof search in the first-order
logic requires us to find out what are the actual resources the proof
search should work with. We observe here that the proof search
process\,---\,as it is the case of the propositional intuitionistic
logic\,---\,can be restricted to formulas that occur only as
subformulas in the initial formula. Of course this time we have to
take into account first-order variables. The following proposition,
which we know how to prove for long normal forms only, sets the
observation in precise terms.
\begin{proposition}
  \label{prop:subformula-property}
  Consider a derivation of $\vdash M:\varphi$ such that $M$ is in the
  long normal form. Each judgement $\Gamma\vdash N:\psi$ that occurs
  in this derivation has the property that for each formula $\xi$ in
  $\Gamma$ and for $\psi$ there is a subformula $\xi'$ of $\varphi$
  such that $\xi = \xi'[X_1:=Y_1,\ldots,X_n:=Y_n]$ where
  $\FV{\xi'} = \{ X_1,\ldots, X_n\}$ and $Y_1,\ldots, Y_n$ are some
  first-order variables.
\end{proposition}
\begin{proof}
  Induction over the size of the term $N$. The details are left to the
  reader.\qed
\end{proof}

We can generalise the property expressed in the proposition above and
say that a formula $\psi$ \emph{emerged} from $\varphi$ when there is
a subformula $\psi_0$ of $\varphi$ and a~substitution
$[X_1:=Y_1,\ldots,X_n:=Y_n]$ with $\FVf{\psi_0}=\{X_1,\ldots,X_n\}$
such that $\psi = \psi_0[X_1:=Y_1,\ldots,X_n:=Y_n]$. We say that a
context $\Gamma$ emerged from $\varphi$ when for each its element
$x:\psi$ the formula $\psi$ emerged from $\varphi$.

\section{Arcadian Automata}
\label{sec:automata}

Our \emph{Arcadian automaton}\footnote{The name Arcadian automata
  stems from the fact that a slightly different and weaker notion of
  \emph{Eden automata} was developed before \cite{RPQNisNE} to deal
  with the fragment of the first-order intuitionistic logic with
  $\forall$ and $\to$ and in which the universal quantifier occurs
  only on positive positions.}  $\aut{A}$ is defined as a tuple
$\< \mathcal A, Q, q^0, \varphi^0, \mathcal I, i, \fv\>$, where
\begin{itemize}
\item $\mathcal A = \< A, \le\>$ is a finite tree, which formally
  describes a division of the automaton control into
  intercommunicating modules; the root of the tree is written
  $\varepsilon$; since the tree is finite we have the relation
  $\rho\succc\rho'$ when $\rho\leq \rho'$ and there is no
  $\rho''\not=\rho$ and $\rho''\not=\rho'$ such that
  $\rho\leq\rho''\leq\rho'$;
\item $Q$ is the set of states; 
\item $q^0\in Q$ is the \emph{initial} state of the automaton;
\item $\varphi^0\in A$ is the \emph{initial} tree node of the automaton;
\item ${\cal I}$ is the set of all instructions;
\item $i\colon Q\to \mathcal{P}(\mathcal{I}) $ is a function which
  gives the set of instructions \emph{available} in a~given state; the
  function $i$ must be such that every instruction belongs to exactly
  one state;
\item $\fv:A\to P(A)$ is a function that describes the binding, it
  has the property that for each node $v$ of $A$ it holds that
  $\fv(v) = \bigcup_{w\in B}\fv(w)$ where
  $B = \{ w \mid v\succc w\}$.
\end{itemize}

Each state may be either existential or universal and 
belongs to an element $a\in A$, so $Q=Q^\exists \cup Q^\forall$, and
$Q^\forall = \bigcup_{a\in A} Q_a^\forall$ and
$Q^\exists = \bigcup_{a\in A} Q_a^\exists$.  The set of states $Q$ is
divided into two disjoint sets $Q_\forall$ and $Q_\exists$ of,
respectively, universal and existential states.

\paragraph{Operational semantics of the automaton.}

An \emph{instantaneous description} (ID) of $\aut{A}$ is a tuple
$\< q, \kappa, w, w', S, V\>$ where
\begin{itemize}
\item $q\in Q$ is the current state, 
\item $\kappa$ is the current node in $A$,
\item $w:A\parfunc V$ is an interpretation of bindings associated with
  $\kappa$ by $\fv(\kappa)$, in particular we require here that
  $\fv(\kappa)\subseteq\dom{w}$,
\item $w':A\parfunc V$ is an auxiliary interpretation of bindings that
  can be stored in this register location of the ID to help implement
  some operations,
\item $S$ is a set called \emph{store}, which contains pairs
  $\<\rho, v\>$ where $\rho\in A$ and $v:A\parfunc V$, we require here
  that $\fv(\rho)\subseteq\dom{v}$,
\item $V$ is the working domain of the automaton.
\end{itemize}
The initial ID is $\<q^0, \varphi^0, \emptyset, \emptyset, \emptyset,
\emptyset\>$.

Intuitively speaking the automaton works as a device which discovers
the knowledge accumulated in the tree ${\cal A}$. It can find new
items of interest in the domain of the discourse and these are stored
in the set $V$ while the facts concerning the elements of $V$ are
stored in $S$. Traditionally, the control of the automaton is
represented by the current state $q$, which belongs to a module
indicated by $\kappa$. We can imagine the automaton as a device that
tries to check if a particular piece of information encoded in the
tree ${\cal A}$ is correct. In this view the piece of information,
which is being checked for correctness at a given point, is represented
by the current node $\kappa$ combined with its interpretation of
bindings~$w$. The interpretation of bindings $w'$ is used to
temporarily hold an interpretation of some bindings.

We have \ref{en:last-number} kinds of instructions in our automata. We
give here their operational semantics. Let us assume that we are in a
current ID $\<q,\kappa, w, w', S, V\>$.  The operation of the instructions
is defined as follows, where we assume $q'\in Q$, $\rho,\rho'\in A$.
\begin{enumerate}
\item $q: \store \rho, \rho' q'$ turns the current ID into\\
  $\<q',\rho', w, \emptyset, S\cup\{\<\rho, (w'\oplus w)|_{\fv(\rho)}\>\}, V\>$,
\item $q: \jmp \rho, q'$ turns the current ID into
  $\<q',\rho, w'', \emptyset, S, V\>$, where\\ $(w'\oplus w)|_{\fv(\kappa)}\subseteq w''$ and
  $\fv(\rho)\subseteq\dom{w''}$,
\item $q:\new \rho, q'$ turns the current ID into
  $\<q',\rho, w, \emptyset, S, V\cup\{X\}\>$, where $X\not\in V$,
\item $q:\check \rho, \rho', q'$ turns the current ID into
  $\<q, \rho', w, \emptyset, S, V\>$, the instruction is applicable only when
  an additional condition is met that there is a pair $\<\rho, v\>\in S$
  such that $v(\rho) = w(\kappa)$,
\item $q:\instl \rho, \rho', q'$ turns the current ID into
  $\<q', \rho', w, \emptyset, S\cup\{\<\rho, w''|_{\fv(\rho)}\>\}, V\cup\{X\}\>$, the instruction is
  applicable only when an additional condition is met
  that there is a node $\rho''\in A$ such that
  $\rho''\succc \rho$ and $w'' = ([\rho'':=X]\oplus w')\oplus w$
  and $X\not\in V$,
\item $q:\instr \rho, q'$ turns the current ID into
  $\<q', \rho, w'', \emptyset, S, V\>$, the instruction is applicable only when
  an additional condition is met that $\kappa\succc\rho$ and
  $w'' = [\gamma:=X]\oplus w|_{\fv(\rho)}$, where
  $\gamma\in\fv(\rho)\backslash\fv(\kappa)$, and $X\in V$,
\item $q:\load \rho, q'$ turns the current ID into
  $\<q',\rho, w'', v, S, V\>$, where\\
  $(w'\oplus w)|_{\fv(\kappa)}\subseteq w''$ and
  $\fv(\rho)\subseteq\dom{w''}$, and $v:A\parfunc V$.
  \label{en:last-number}
\end{enumerate}
These instructions abstract the basic operations associated with the
process of proving in predicate logic.  Observe that the content of
the additional register loaded by the instruction $\load$ can be used
only for the immediately following instruction as all the other
instructions erase the content of the register. 

It is also interesting to observe that the set of instructions contains
in addition to standard assembly-like instructions two instructions
$\instl$, $\instr$ that deal with pattern instantiation.

\newcounter{rcount}
\setcounter{rcount}{1}
\makeatletter
\newcommand{\rnumber}[1]{%
(\arabic{rcount})%
\protected@write \@auxout {}{\string
  \newlabel{#1}{{\arabic{rcount}}{}{}{}{}}}%
\stepcounter{rcount}}

\makeatother

\begin{figure}[tb]
  \centering
  \begin{displaymath}
    \begin{array}{@{}ll@{\;\;\;}l@{}}
      \multicolumn{3}{c}{\mbox{\textbf{Structural decomposition instructions}}} \\
      \hline
      \hline\\
      \rnumber{ins:to}&
       \varphi_1\to\varphi_2 
      & q_{\varphi_1\to\varphi_2}^\forall:\store \varphi_1,\varphi_2,
        q_{\varphi_2}^\exists\\
    & & \Rightarrow\!
        \<q_{\varphi_1\to\varphi_2}^\forall,\varphi_1\!\to\!\varphi_2,w,\emptyset,S,V\>\!\to\!%
        \<q_{\varphi_2}^\exists,\varphi_2,w,\emptyset,S\cup\{\<\varphi_1,w|_{\fv(\varphi_1)}\>\},V\>
      \\[0.5ex]
      \hline\\[-1.5ex]
      \rnumber{ins:conj}&
      \varphi_1\land\varphi_2 
       & q_{\varphi_1\land\varphi_2}^\forall:\jmp \varphi_1,
        q_{\varphi_1}^\exists\\[0.5ex]
     & & \Rightarrow\;\;
        \<q_{\varphi_1\land\varphi_2}^\forall,\varphi_1\land\varphi_2,w,\emptyset,S,V\>\to%
        \<q_{\varphi_1}^\exists,\varphi_1,w,\emptyset,S,V\>\\[0.5ex]

     & & q_{\varphi_1\land\varphi_2}^\forall:\jmp \varphi_2,
        q_{\varphi_2}^\exists\\
     & & \Rightarrow\;\;
        \<q_{\varphi_1\land\varphi_2}^\forall,\varphi_1\land\varphi_2,w,\emptyset,S,V\>\to%
        \<q_{\varphi_2}^\exists,\varphi_2,w,\emptyset,S,V\>\\[0.5ex]

      \hline\\[-1.5ex]
      \rnumber{ins:disI}&
         \varphi_1\lor\varphi_2
       & q_{\varphi_1\lor\varphi_2}^\exists:\jmp \varphi_1,
        q_{\varphi_1}^\exists\\[0.5ex]
     & & \Rightarrow\;\;
        \<q_{\varphi_1\lor\varphi_2}^\exists,\varphi_1\lor\varphi_2,w,\emptyset,S,V\>\to%
        \<q_{\varphi_1}^\exists,\varphi_1,w,\emptyset,S,V\>\\[0.5ex]

     & & q_{\varphi_1\lor\varphi_2}^\exists:\jmp \varphi_2,
        q_{\varphi_2}^\exists\\[0.5ex]
     & & \Rightarrow\;\;
        \<q_{\varphi_1\lor\varphi_2}^\forall,\varphi_1\lor\varphi_2,w,\emptyset,S,V\>\to%
        \<q_{\varphi_2}^\exists,\varphi_2,w,\emptyset,S,V\>\\[0.5ex]

      \hline\\[-1.5ex]
      \rnumber{ins:forall}&
      \forall X.\varphi 
       & q_{\forall X.\varphi}^\forall:\new \varphi,
        q_\varphi^\exists\\[0.5ex]

     & & \Rightarrow\;\;
        \<q_{\forall X.\varphi}^\forall,\forall X.\varphi,w,\emptyset,S,V\>\to%
        \<q_{\varphi}^\exists,\varphi,[\forall X.\varphi:=Y]\oplus w,\emptyset,S,V\cup\{Y\}\>\\
     & & \qquad \mbox{where $Y\not\in V$}\\[0.5ex]

      \hline\\[-1.5ex]
      \rnumber{ins:exists}&
      \exists X.\varphi 
       & q_{\exists X.\varphi}^\forall:\instr \varphi,
        q_\varphi^\exists\\[0.5ex]

     & & \Rightarrow\;\;
        \<q_{\exists X.\varphi}^\forall,\exists X.\varphi,w,\emptyset,S,V\>\!\to\!%
        \<q_{\varphi}^\exists,\varphi,[\exists X.\varphi:=Y]\oplus w|_{\fv(\exists X.\varphi)},\emptyset,S,V\>\\
     & & \qquad\mbox{where $Y\in V$}
    \end{array}
  \end{displaymath}
  \caption{Structural decomposition instructions of the automaton}
  \label{fig:structural-instructions}
\end{figure}

The following notion of acceptance is defined inductively.  We say
that the automaton $\aut{A}$ \emph{eventually accepts} from an ID
$a=\<q, \kappa, w, w', S, V\>$ when
\begin{itemize}
\item $q$ is universal and there are no instructions available in
  state $q$ (i.e. $i(q)=\emptyset$, such states are called
  \emph{accepting states}), or
\item $q$ is universal and, for each instruction $i$ available in $q$,
  the automaton started in an ID $a'$ eventually accepts,
  where $a'$ is obtained from $a$ by executing $i$,
\item if $q$ is existential and, for some instruction $i$ available in
  state $q$ the automaton started in an ID $a'$ eventually
  accepts, where $a'$ is obtained from $a$ by executing $i$.
\end{itemize}

The definition above actually defines inductively a certain kind of
tree, the nodes of which are IDs and children of a node are determined
by the configurations obtained by executing of available instructions.
Actually, we can view the process described above not only as a
process of reaching acceptance, but also as a process of recognising
of the tree.  In this light the automaton is eventually accepting from
an initial configuration if the language of its `runs' is not empty.
As a result we can talk about the acceptance of such automata by
referring to the \emph{emptiness problem}.

Here is a basic monotonicity property of the automata.

\begin{proposition}
  \label{prop:basic}
  If the automaton $\aut{A}$ eventually accepts from
  $\<q,\kappa, w, w', S, V\>$ and $w\subseteq w''$ then the automaton $\aut A$
  eventually accepts from $\<q,\kappa, w'', w', S, V\>$.
\end{proposition}
\begin{proof}
  Induction over the definition of the configuration from which
  automaton eventually accepts. The details are left to the reader.\qed
\end{proof}

\subsection{From formulas to automata}

We can now define an Arcadian automaton
$\aut{A}_{\varphi}= \<\mathcal A, Q, q_\varphi^\exists, 
\varphi, \mathcal I, i, \fv\>$
that corresponds to provability of the formula $\varphi$. For
technical reasons we assume that the formula is closed. This
restriction is not essential since the provability of a~formula with
free variables is equivalent to the provability of its universal closure.
The components of the automaton are as follows.

\begin{itemize}
  \item ${\mathcal A} = \< A,\leq \>$ is the syntax tree of the formula $\varphi$.
  \item $Q = \{ q_{\psi}^\forall, q_{\psi}^\exists, 
              q_{\psi,\lor}^\forall, q_{\psi,\to}^\forall,
              q_{\psi,\exists}^\forall, q_{\psi,\bot}^\forall
              \mid \mbox{ for all subformulas } \psi \mbox{ of }
              \varphi\}$. The states annotated with the superscript
              $\forall$ belong to $Q^\forall$ while the states with
              the superscript $\exists$ belong to $Q^\exists$.
  \item $q_\varphi^\exists$ is the initial state (which
    means the goal of the proving process is $\varphi$).
  \item The initial state and initial tree node are
    $q_\varphi^\exists$ and $\varphi$, respectively.
  \item ${\mathcal I}$ and $i$ are presented in
    Fig.~\ref{fig:structural-instructions} and
    \ref{fig:instructions}. We describe them in more detail 
    below.
  \item $\fv:A\to P(A)$ is defined so that  $\fv(\psi) =\{ \bind_\varphi(\psi,X) \mid
    X\in\FV{\psi}\}$.
\end{itemize}

\begin{figure}[htbp]
  \centering
  \begin{displaymath}
    \begin{array}{ll@{\quad}l}
      \multicolumn{3}{c}{\mbox{\textbf{Non-structural instructions}}}\\
      \hline
      \hline\\
      \rnumber{ins:toforall}&
        & q_\varphi^\exists:\jmp \varphi, q_\varphi^\forall\\ 
      & & \Rightarrow\;\;
        \<q_{\varphi}^\exists,\varphi,w,\emptyset,S,V\>\to%
        \<q_{\varphi}^\forall,\varphi,w,\emptyset, S,V\>\\[0.5ex]

      \rnumber{ins:conjE}&
        & q_{\varphi_i}^\exists:\jmp \varphi_1\land\varphi_2,
          q_{\varphi_1\land\varphi_2}^\exists \mbox{ for } i=1,2\\
      & & \Rightarrow\;\;
        \<q_{\varphi}^\exists,\varphi,w,\emptyset,S,V\>\to%
        \<q_{\varphi_1\land\varphi_2,\land}^\forall,\varphi_1\land\varphi_2,w'',\emptyset,S,V\>\\[0.5ex]

      \rnumber{ins:orleft}&
        & q_\varphi^\exists:\load \varphi, q_{\varphi,\lor}^\forall\\
      & & \Rightarrow\;\;
        \<q_{\varphi}^\exists,\varphi,w,\emptyset,S,V\>\to%
        \<q_{\varphi,\lor}^\forall,\varphi,w,w',S,V\>\\[0.5ex]

      \rnumber{ins:toleft}&
        & q_\varphi^\exists:\jmp \varphi, q_{\varphi,\to}^\forall\\
      & & \Rightarrow\;\;
        \<q_{\varphi}^\exists,\varphi,w,\emptyset,S,V\>\to%
        \<q_{\varphi,\to}^\forall,\varphi,\hat{w},\emptyset,S,V\> \quad
          \mbox{ where } w\subseteq \hat{w}\\[0.5ex]

      \rnumber{ins:forallE}&
        & q_\varphi^\exists:\jmp \forall X.\varphi, q_{\forall X.\varphi}^\exists\\
      & & \Rightarrow\;\;
        \<q_{\varphi}^\exists,\varphi,w,\emptyset,S,V\>\to%
        \<q_{\forall X.\varphi}^\exists,\forall X.\varphi,w,\emptyset,S,V\>\\[0.5ex]

      \rnumber{ins:existsE}&
        & q_\varphi^\exists:\load \varphi, q_{\varphi,\exists}^\forall\\
      & & \Rightarrow\;\;
        \<q_{\varphi}^\exists,\varphi,w,\emptyset,S,V\>\to%
        \<q_{\varphi,\exists}^\forall,\varphi,\hat{w},w',S,V\>\\[0.5ex]

      \rnumber{ins:bot}&
        & q_\varphi^\exists:\jmp \varphi, q_{\varphi,\bot}^\forall\\
      & & \Rightarrow\;\;
        \<q_{\varphi}^\exists,\varphi,w,\emptyset,S,V\>\to%
        \<q_{\varphi,\bot}^\forall,\varphi,w,\emptyset,S,V\>\\[0.5ex]

      \hline\\[-1.5ex]
      \rnumber{ins:var}&
        & q_{\varphi}^\exists:\check
        \varphi,\varphi,q_{\axiom}^\forall\\[0.5ex]
      & & \Rightarrow\;\;
        \<q_{\varphi}^\exists,\varphi,w,\emptyset,S,V\>\to%
        \<q_{\varphi}^\exists,\varphi,w,\emptyset,S,V\>\\[0.5ex]
      \hline\\[-1.5ex]

      \rnumber{ins:orEor}& 
       & q_{\varphi,\lor}^\forall: \jmp
        \psi_1\lor\psi_2,q_{\psi_1\lor\psi_2}^\exists\\
      & & \Rightarrow\;\;
        \<q_{\varphi,\lor}^\forall,\varphi,w,w',S,V\>\to%
        \<q_{\psi_1\lor\psi_2}^\exists,\psi_1\lor\psi_2,w',\emptyset,S,V\>\\[0.5ex]

      \rnumber{ins:orEone}&
       & q_{\varphi,\lor}^\forall: \store
        \psi_1,\varphi,q_{\varphi}^\exists\\
      & & \Rightarrow\;\;
        \<q_{\varphi,\lor}^\forall,\varphi,w,w',S,V\>\to%
        \<q_{\varphi}^\exists,\varphi,w',\emptyset,S',V\>\\[0.5ex]
      & & \mbox{where } S' = S\cup\{\<\psi_1,w'|_{\fv(\psi_1)}\>\}\\[0.5ex]

      \rnumber{ins:orEtwo}&
       & q_{\varphi,\lor}^\forall: \store
        \psi_2,\varphi,q_{\varphi}^\exists\\
      & & \Rightarrow\;\;
        \<q_{\varphi,\lor}^\forall,\varphi,w,w',S,V\>\to%
        \<q_{\varphi}^\exists,\varphi,w',\emptyset,S',V\>\\[0.5ex]
      & & \mbox{where } S' = S\cup\{\<\psi_2,w'|_{\fv(\psi_2)}\>\}\\[0.5ex]
      \multicolumn{3}{c}{$\scriptsize (\ref{ins:orEone}) and (\ref{ins:orEtwo}) should
      be instantiated with $\psi_1$ and $\psi_2$'s which were used in 
      (\ref{ins:orEor}$).}\\
      \hline\\[-1.5ex]

      \rnumber{ins:toEleft}&
       & q_{\varphi,\to}^\forall: \jmp
        \psi\to\varphi,q_{\psi\to\varphi}^\exists\\
      & & \Rightarrow\;\;
        \<q_{\varphi,\to}^\forall,\varphi,w,\emptyset,S,V\>\to%
          \<q_{\psi\to\varphi}^\exists,\psi\to\varphi,w,\emptyset,S,V\>\\[0.5ex]

      \rnumber{ins:toEright}&
       & q_{\varphi,\to}^\forall: \jmp
        \psi,q_{\psi}^\exists\\[1ex]
      & & \Rightarrow\;\;
        \<q_{\varphi,\to}^\forall,\varphi,w,\emptyset,S,V\>\to%
          \<q_{\psi}^\exists,\psi,w,\emptyset,S,V\>\\[0.5ex]
      \multicolumn{3}{c}{$\scriptsize (\ref{ins:toEright}) should
      be instantiated with $\psi$ and $\varphi$'s which were used in 
      (\ref{ins:toEleft}$).}\\
      \hline\\[-1.5ex]

      \rnumber{ins:existsEone}& 
       & q_{\varphi,\exists}^\forall: \jmp
        \exists X.\psi,q_{\exists X.\psi}^\exists\\
      & & \Rightarrow\;\;
        \<q_{\varphi,\exists}^\forall,\varphi,w,w',S,V\>\to%
        \<q_{\exists X.\psi}^\exists,\exists X.\psi,w',\emptyset,S',V\>\\[0.5ex]

      \rnumber{ins:existsEtwo}&
       & q_{\varphi,\exists}^\forall: \instl
        \psi,\varphi,q_{\varphi}^\exists\\
      & & \Rightarrow\;\;
        \<q_{\varphi,\exists}^\forall,\varphi,w,w',S,V\>\to%
        \<q_{\varphi}^\exists,\varphi,w,\emptyset,S',V\>\\[0.5ex]
     & & \mbox{where } w'' = ([\exists X.\psi:=X]\oplus w')\oplus w, 
         S' = S\cup\{\<\psi,w''|_{\fv(\psi)}\>\}\\[0.5ex]
      \multicolumn{3}{c}{$\scriptsize (\ref{ins:existsEone}) should
      be instantiated with $\psi$ and $\varphi$'s which were used in 
      (\ref{ins:existsEtwo}$).}\\
      \hline\\[-1.5ex]

      \rnumber{ins:botE}&
      & q_{\varphi,\bot}^\forall:\jmp \bot,q_\bot^\exists\\
      & & \Rightarrow\;\;
        \<q_{\varphi,\bot}^\forall,\varphi,w,\emptyset,S,V\>\to%
        \<q_{\varphi,\bot}^\forall,\bot,w,\emptyset,S,V\>\\[0.5ex]
    \end{array}
  \end{displaymath}

  \caption{Non-structural instructions of the automaton}
  \label{fig:instructions}
\end{figure}

Fig.~\ref{fig:structural-instructions} and~\ref{fig:instructions}
present the patterns of possible instructions in ${\cal I}$.  Each of
the instruction patterns starts with a state of the form
$q^\triangledown_\psi$ or of the form $q^\triangledown_{\psi,\bullet}$
where $\triangledown$ is a quantifier ($\forall$ or $\exists$), $\psi$
is a subformula of $\varphi$ and $\bullet$ is one of the symbols
$\lor, \to, \bot,\exists$. For each of the patterns we assume
${\cal I}$ contains all the instructions that result from
instantiating the pattern with all possible subformulas that match the
form of $\psi$ (e.g.\ in case $\psi = \psi_1\to\psi_2$ we take all the
subformulas with $\to$ as the main symbol). The function
$i:Q\to P({\cal I})$ is defined so that for a state $q^Q_\psi$ it
returns all the instructions which start with the state. In addition to
the instructions they present the way a configuration is transformed
by each of the instructions. This serves to facilitate understanding
the proofs.

As the figure suggests, the instructions of the automaton can be
divided into two groups\,---\,structural decomposition instructions
and non-structural ones.  The structural instructions are used to
decompose a formula into its structural subformulas. On the left-hand
side of each of the structural instructions we present the formula
the instruction decomposes. The other rules represent operations that
manipulate other elements of configuration with possible change of the
goal formula, see example below for illustration.

\paragraph{Example.} Consider the formula $\varphi=\forall x(P(x))\to\forall y\exists x
P(x)$. In order to build the Arcadian automaton for that formula first we have
to build the tree $A$ of it, which is shown in Fig.~\ref{fig:syntree}.\hfil~
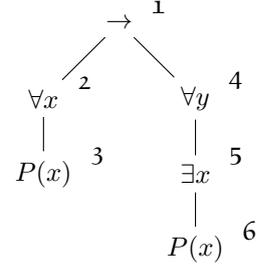
\begin{wrapfigure}[7]{r}[-0em]{0.36\textwidth}

\vspace{-9ex}

\centering
\begin{tikzpicture}
	[level distance = 1cm]
	\node (is-root) { $\to $}
		[sibling distance=2cm]
		child { node (two){$\forall x $} 
			child { node (three) { $P(x)$ } }
		}
		child {
			node (four){$\forall y$}
				child { node (five){$\exists x$}
					child { node (six){$P(x)$} }
				}
		};
	\node[anchor=south west] at (is-root.east) {$\stateno 1$};
	\node[anchor=south west] at (two.east) {$\stateno 2$};
	\node[anchor=south west] at (three.east) {$\stateno 3$};
	\node[anchor=south west] at (four.east) {$\stateno 4$};
	\node[anchor=south west] at (five.east) {$\stateno 5$};
	\node[anchor=south west] at (six.east) {$\stateno 6$};
\end{tikzpicture}\\[-1ex]
\caption{Syntax tree of the formula.}
\label{fig:syntree}
\end{wrapfigure}

\noindent
The instructions available ($\mathcal I$) are:\\[1ex]
\begin{tabular}{rlllrlllrlll}
 (1) \iline \forall 1 {\store \stateno 2, \stateno 4, \stateid\exists 4} (19) \ixline \forall 1 \exists {\jmp \stateno 5, \stateid\exists 5}			\\
 (4) \iline \forall 2 {\new \stateno 3, \stateid\exists 3}		 (19) \ixline \forall 4 \exists {\jmp \stateno 5, \stateid\exists 5}			\\
 (4) \iline \forall 4 {\new \stateno 5, \stateid\exists 5}		 (19) \ixline \forall 5 \exists {\jmp \stateno 5, \stateid\exists 5}			\\
 (5) \iline \forall 5 {\instr \stateno 6, \stateid\exists 6}		 (20) \ixline \forall 1 \exists {\instl \stateno 5, \stateno 1, \stateid\exists 1}	\\
(10) \iline \exists 3 {\jmp \stateno 2, \stateid\exists 2}		 (20) \ixline \forall 4 \exists {\instl \stateno 5, \stateno 4, \stateid\exists 4}	\\
(10) \iline \exists 6 {\jmp \stateno 2, \stateid\exists 2}		 (20) \ixline \forall 5 \exists {\instl \stateno 5, \stateno 5, \stateid\exists 5}	\\
\end{tabular}\\[-0.5ex]

\noindent 
the instructions available for any $a\in A$~are\\[-0.5ex]

\noindent
\begin{tabular}{rlllrlllrlll}
 (6)&\ailine \exists	  {\jmp a,	q^\forall_a} &
 (8)&\ailine \exists	  {\load a,	q^\forall_{a, \lor}} &
 (9)&\ailine \exists	  {\jmp a,	q^\forall_{a,\to}} \\
(11)&\ailine \exists	  {\load a,	q^\forall_{a,\exists}} &
(12)&\ailine \exists	  {\jmp a,	q^\forall_{a,\bot}} &
(13)&\ailine \exists	  {\check a, a, q^\forall_\axiom} \\
(21)&\aixline \forall \to {\jmp \bot,	q^\exists_\bot} \\
\end{tabular}\\[-0.5ex]

\noindent
The set of states can be easily written using the definition. To
calculate $\fv$ we need to calculate $\bind$s first. We have
$\bind_{\stateno 1} (\stateno 3, x) = \stateno 2$ and
$\bind_{\stateno 1} (\stateno 6, x) = \stateno 5$; therefore
$\fv(\stateno 3) = \left\{ \stateno 2 \right\}$,
$\fv(\stateno 6) = \left\{ \stateno 5 \right\}$ and
$\fv(\textrm{otherwise}) = \emptyset$. $q^0 = \stateid\exists 1$ and
$\varphi^0=\varphi$. The initial ID is $q=q^{\exists}_{\stateno 1}$,
$\kappa=\stateno 1$, and the other elements of the description are
empty sets. A successful run of the automaton is as follows: 
$\jmp \stateno 1,\stateid{\forall}{1}$ (rule (6), initial instruction
leads to the structural decomposition of the main connective $\to$);
$\store \stateno 2, \stateno 4, \stateid\exists 4$ (r. (1), as the
result of the decomposition, the formula at the node $\stateno 2$ is
moved to the context, and the formula at $\stateno 4$ becomes the
proof goal);
$\jmp \stateno 4,\stateid\forall 4$ (r. (6), 
we progress to the structural decomposition of $\forall$);
$\new \stateno 5,\stateid\exists 5$ (r. (4),
we introduce fresh eigenvariable, say $X_1$, for the universal
quantifier);
$\jmp \stateno 5,\stateid\forall 5$ (r. (6), 
we progress to the structural decomposition of $\exists$);
$\instr \stateno 6,\stateid\exists 6$ (r. (5),
we produce a witness for the existential quantifier, which can be just
$X_1$);
$\jmp \stateno 2,\stateid\exists 2$ (r. (10), we progress now with the
non-structural rule that handles instantiation of the universal
assumption from the node $2$); 
and now we can conclude with 
$\check \stateno 2, \stateno 2, q_{\axiom}$ (r. (13)) that
directly leads to acceptance.

\paragraph{From derivability questions to IDs} A proof search process
in the style of Ben-Yelles~\cite{BenYelles79} works by solving
derivability questions of the form $\Gamma\vdash ?:\psi$.  We relate
this style of proof search to our automata model by a translation of
such a~question into an ID of the automaton. Suppose that the initial
closed formula is $\varphi$. We define the configuration of
$\aut{A}_\varphi$ that corresponds to $\Gamma\vdash ?:\psi$ by
exploiting the conclusion of
Proposition~\ref{prop:subformula-property}. This proposition makes it
possible to associate a substitution $w_\psi$ with $\psi$ and $w_\xi$ with
each assignment $x:\xi\in\Gamma$. The resulting configuration is
  $a_{\Gamma,\psi} = \<q_{\psi_0}^\exists,\psi_0, w_\psi,\emptyset, S_{\Gamma,\psi},V_{\Gamma,\psi}\>$
where $S_{\Gamma,\psi} = \{ \<\xi,\psi_\xi\> \mid x:\xi\in\Gamma \}$
and $V_{\Gamma,\psi} = \FVf{\Gamma,\psi}$ as well as
$w_\psi(\psi_0) = \psi$.

\begin{lemma}
  \label{lemma:judgments-to-configurations}
  If $\Gamma\vdash M:\psi$ is derivable and such that $\Gamma$ and
  $\psi$ emerged from $\varphi$ then $\aut{A}_{\varphi}$ eventually
  accepts from the  configuration
  $\<q_{\psi_0}^\exists,\psi_0, w_\psi,\emptyset,
  S_{\Gamma,\psi},V_{\Gamma,\psi}\>$.
\end{lemma}
\begin{proof}
  We may assume that $M$ is in the long normal form.  The proof is by
  induction over the derivation of $M$.  We give here only the
  most interesting cases.

  If the last rule is $(var)$,  we can apply the
  instruction (\ref{ins:var}) that checks if the
  formula $w_{\psi}(\psi_0)$ is in  $S_{\Gamma,\psi}$. Then
  the resulting state $q_{\axiom}^\forall$ is an accepting state.

  If the last rule is the $(\land I)$ rule then
  $\psi = \psi_1\land \psi_2$ and we have shorter derivations for
  $\Gamma\vdash M_1:\psi_1$ and $\Gamma\vdash M_2:\psi_2$, which by
  induction hypothesis give that $\aut{A}_\varphi$ eventually accepts
  from the configurations
  $\<q_{\psi_{i0}}^\exists,\psi_{i0}, w_{\psi_i},\emptyset,
  S_{\Gamma,\psi_i},V_{\Gamma,\psi_i}\>.$ for $i=1,2$ where we note
  that $w_{\psi_i}=w_\psi$, $S_{\Gamma,\psi_i} = S_{\Gamma,\psi}$ and
  $V_{\Gamma,\psi_i} = V_{\Gamma,\psi}$.  We can now use the rule
  (\ref{ins:toforall}) to turn the existential state
  $q_{\psi}^\exists$ into the universal one $q_{\psi}^\forall$ for which
  there are two instructions available in (\ref{ins:conj}), and these
  turn the current
  configuration into the corresponding above mentioned ones.

  If the last rule is the $(\land Ei)$ rule for $i=1,2$ then we know
  that $\psi=\psi_i$ for one of $i=1,2$ and
  $\Gamma\vdash M':\psi_1\land\psi_2$ is derivable through a shorter
  derivation, which means by the induction hypothesis that
  $\aut{A}_\varphi$ eventually accepts from the configuration
  $\<q_{\psi_1\land\psi_2}^\exists,\psi_{1}\land\psi_2,
  w_{\psi_1\land\psi_2},\emptyset,
  S_{\Gamma,\psi_1\land\psi_2},V_{\Gamma,\psi_1\land\psi_2}\>$ where
  actually $w_{\psi_1\land\psi_2}|_{\fv(\psi_i)}\subseteq w_{\psi_i}$
  and $\fv(\psi_i)\subseteq\dom{w_{\psi_1\land\psi_2}}$ for both
  $i=1,2$. Moreover, $S_{\Gamma,\psi_1\land\psi_2} = S_{\Gamma,\psi}$ and
  $V_{\Gamma,\psi_1\land\psi_2} = V_{\Gamma,\psi}$. This configuration
  can be obtained from the current one using respective instruction
  presented at~(\ref{ins:conjE}).

  If the last rule is the $(\to E)$ rule then we have shorter
  derivations for $\Gamma\vdash M_1: \psi'\to \psi$ and
  $\Gamma\vdash M_2: \psi'$. The induction hypothesis gives that
  $\aut{A}_\varphi$ eventually accepts from the configurations
  \begin{displaymath}
    \begin{array}{l}
      \<q_{\psi'_0\to\psi_0}^\exists,\psi'_0\to\psi_0, w_{\psi'\to\psi},\emptyset,
      S_{\Gamma,\psi'\to\psi},V_{\Gamma,\psi'\to\psi}\>,\\
      \<q_{\psi'_0}^\exists,\psi'_0, w_{\psi'},\emptyset, S_{\Gamma,\psi'},V_{\Gamma,\psi'}\>.
    \end{array}
  \end{displaymath}Note that actually
  $S_{\Gamma,\psi'\to\psi}= S_{\Gamma,\psi}$ and
  $V_{\Gamma,\psi'\to\psi}=V_{\Gamma,\psi}$.  We can now use
  the instruction (\ref{ins:toleft}) to turn the current
  configuration into
  \begin{displaymath}
    \<q_{\psi_0,\to}^\forall,\psi_0, w_{\psi'\to\psi},\emptyset, S_{\Gamma,\psi},V_{\Gamma,\psi}\>,
  \end{displaymath}
  which can be turned into the desired two configurations with the
  instructions (\ref{ins:toEleft}) and (\ref{ins:toEright})
  respectively.

  If the last rule is the $(\forall I)$ rule then
  $\psi = \forall X. \psi_1$ and we have a shorter derivation for
  $\Gamma\vdash M_1:\psi_1$ (where $X$ is a fresh variable by the
  eigenvariable condition), which by the induction hypothesis gives
  that $\aut{A}_\varphi$ eventually accepts from the configuration
  $$
  \<q_{\psi_{10}}^\exists,\psi_{10}, w_{\psi_1},\emptyset,
  S_{\Gamma,\psi_1},V_{\Gamma,\psi_1}\>,
  $$ 
  where $w_{\psi_1}(\psi_{10}) = \psi_1$,
  $S_{\Gamma,\psi_1} = S_{\Gamma,\psi}$ and
  $V_{\Gamma,\psi_1} = V_{\Gamma,\psi}\cup\{X\}$.

  We observe now that the instruction (\ref{ins:toforall}) transforms
  the current configuration to
  $\<q_{\forall X.\psi_{10}}^\forall,\psi, w_\psi,\emptyset,
  S_{\Gamma,\psi},V_{\Gamma,\psi}\>$ and then the $\new$ instruction from
  (\ref{ins:forall}) adds appropriate element to $V_{\Gamma,\psi}$ and
  turns the configuration into the awaited one.

  If the last rule is the $(\exists E)$ rule then we know that
  $\Gamma\vdash M_1:\exists X.\psi_1$ and
  $\Gamma,x:\psi_1\vdash M_2:\psi$ are derivable through shorter
  derivations, which means by the induction hypothesis that
  $\aut{A}_\varphi$ eventually accepts from configurations
  \begin{equation}
    \label{eq:ins-exists}
    \begin{array}{l}
        \<q_{\exists X.\psi_{01}}^\exists,\exists X.\psi_{01},
        w_{\exists X.\psi_1},\emptyset, S_{\Gamma,\exists
        X.\psi_1},V_{\Gamma,\exists X.\psi_1}\>, \\
        \<q_{\psi_{0}}^\exists,\psi_{0},
        w_{\psi},\emptyset, S_{\Gamma',\psi},V_{\Gamma',\psi}\>
    \end{array}
  \end{equation}
  where
  $w_{\exists X.\psi_1}(\exists X.\psi_{01}) = \exists X.\psi_1$,
  $w_\psi(\psi_0)=\psi$ and $\Gamma' = \Gamma, x:\psi_1$, which
  consequently means that
  $S_{\Gamma',\psi} = S_{\Gamma,\psi}\cup\{\<\psi_{01},w'\>\}$ and
  $V_{\Gamma',\psi} = V_{\Gamma,\psi}\cup\{X\}$ where
  $w' = [\exists X.\psi_{01}:=X]\oplus w_{\exists
    X.\psi_1}|_{\fv(\psi_{01})}$.
  Note that $x$ is a fresh proof variable by definition and
  $X$ is a fresh variable by the eigenvariable condition.

  We observe that the current configuration can be transformed to
  $$
  \<q_{\psi_0,\exists}^\forall,\psi_0, w_\psi,w_{\exists X.\psi_1},
  S_{\Gamma,\psi},V_{\Gamma,\psi}\>
  $$
  by the instruction (\ref{ins:existsE}). This in turn is transformed
  to the configurations (\ref{eq:ins-exists}) by instructions
  (\ref{ins:existsEone}) and (\ref{ins:existsEtwo}) respectively.
\qed
\end{proof}

We need a proof in the other direction. To express the statement of
the next lemma we need the notation $\Gamma_S$ for a context
$x_1:w_1(\psi_1),\ldots,x_n:w_n(\psi_n)$ where
$S=\{ \<\psi_1,w_1\>,\ldots,\<\psi_n,w_n\>\}$.

\begin{lemma}
  \label{lemma:configurations-to-judgements}
  If $\aut{A}_{\varphi}$ eventually accepts from the configuration
  $\< q_\psi^\exists, \psi, w, \emptyset, S, V\>$ then there is a proof term $M$
  such that $\Gamma_S\vdash M:w(\psi)$.
\end{lemma}
\begin{proof}
  The proof is by induction over the definition of the eventually
  accepting configuration by cases depending on the currently
  available instructions.  Note that only instructions
  (\ref{ins:disI}), (\ref{ins:toforall}), (\ref{ins:conjE}),
  (\ref{ins:orleft}), (\ref{ins:toleft}), (\ref{ins:forallE}),
  (\ref{ins:existsE}), (\ref{ins:bot}), (\ref{ins:var}) are available
  for states of the form $q_\phi^\exists$.

  We can immediately see that if one of the instructions
  (\ref{ins:disI}) from Fig.~\ref{fig:structural-instructions} is
  used then the induction hypothesis applied to resulting
  configurations brings the assumption of the respective rule
  $(\lor Ii)$ for $i=1,2$ and we can apply it to obtain the
  conclusion.

  Then taking the instruction (\ref{ins:toforall}) moves control to
  one of the instructions present in
  Fig.~\ref{fig:structural-instructions} and these move control to
  configurations from which the induction hypothesis gives the
  assumptions of the introduction rules
  $(\to I), (\land I),$ $(\forall I), (\exists I)$ respectively.

  Next taking the instructions (\ref{ins:orleft}), (\ref{ins:toleft}),
  (\ref{ins:existsE}), (\ref{ins:bot}) move control to further
  non-structural rules in Fig.~\ref{fig:instructions} and these move
  control to configurations from which the induction hypothesis gives
  the assumptions of the elimination rules $(\lor E)$, $(\to E)$,
  $(\exists E)$, and $(\bot E)$. At the same time the instructions
  (\ref{ins:conjE}), (\ref{ins:forallE}), move control directly to
  configurations from which the induction hypothesis gives the
  assumptions of the elimination rules $(\land E)$, $(\forall E)$.

  At last the instruction (\ref{ins:var}) directly represents the use
  of the $(var)$ rule.

  More details of the reasoning can be observed by referring to
  relevant parts in the proof of
  Lemma~\ref{lemma:judgments-to-configurations} and adapting them to
  the current situation.\qed
\end{proof}

\begin{theorem}[Main theorem]
  The provability in intuitionistic first-order logic is equivalent to
  the emptiness problem for Arcadian automata.
\end{theorem}
\begin{proof}
  Let $\varphi$ be a formula of the first-order intuitionistic logic.
  The emptiness problem for $\aut{A}_\varphi$ is equivalent to
  checking if the initial configuration of this Arcadian automaton is
  eventually accepting. This in turn is by
  Lemma~\ref{lemma:judgments-to-configurations} and
  Lemma~\ref{lemma:configurations-to-judgements} equivalent to
  derivability of $\vdash\varphi$.\qed
\end{proof}

\section{Conclusions}
\label{sec:conclusions}

We proposed a notion of automata that can simulate search for proofs
in normal form in the full first-order intuitionistic logic, which can
be viewed by the Curry-Howard isomorphism as a program synthesis for a
simple functional language. This notion enables the possibility to
apply automata theoretic techniques to inhabitant search in this type
system. Although the emptiness problem for such automata is
undecidable (as the logic is, \cite{SchubertUZ2015}), the notion
brings a new perspective to the proof search process which can reveal
new classes of formulae for which the proof search can be made
decidable. In particular this automata, together with earlier
investigations \cite{RPQNisNE,SchubertUZ2015}, bring to the attention
that decidable procedures must constrain the growth of the subset $V$
in ID of automata presented here.

\bibliographystyle{splncs03} 
\bibliography{spojniki}

\end{document}